%% file: main.tex
\documentclass[a4paper,onecolumn,accepted=2022-07-11]{quantumarticle}
\pdfoutput=1

\usepackage[utf8]{inputenc}

\usepackage{amsmath,amssymb,amsfonts,amsthm,mathtools}

\usepackage{booktabs,multirow,grffile}
\usepackage{placeins}
\makeatletter
\AtBeginDocument{%
	\expandafter\renewcommand\expandafter\subsection\expandafter
	{\expandafter\@fb@secFB\subsection}%
	\newcommand\@fb@secFB{\FloatBarrier
		\gdef\@fb@afterHHook{\@fb@topbarrier \gdef\@fb@afterHHook{}}}%
	\g@addto@macro\@afterheading{\@fb@afterHHook}%
	\gdef\@fb@afterHHook{}%
}
\makeatother

\input{braket}
\input{symphony}

\usepackage{tikz}
\usetikzlibrary{quantikz}

\bibliography{bibliography.bib}

\newcommand\field\mathds
\newcommand\op\mathbf
\newcommand\1{\mathds 1}
\DeclareMathOperator{\BigO}{O}
\newcommand\ii{\mathds{i}}
\newcommand\ee{\mathrm{e}}

\usepackage{bm}
\newcommand\bitalpha{A}

\title{\mbox{Fast Black-Box Quantum State Preparation}}
\author{Johannes Bausch\thanks{CQIF, DAMTP, University of Cambridge, UK. \texttt{jkrb2@cam.ac.uk}}}
\date{May 2021\thanks{Revised June 2022}}

\begin{document}

\maketitle

\begin{abstract}
    Quantum state preparation is an important ingredient for other higher-level quantum algorithms, such as Hamiltonian simulation, or for loading distributions into a quantum device to be used e.g.\ in the context of optimization tasks such as machine learning.
    Starting with a generic ``black box'' method devised by \citeauthor{Grover2000}, which employs amplitude amplification to load coefficients calculated by an oracle, there has been a long series of results and improvements with various additional conditions on the amplitudes to be loaded, culminating in \citeauthor{Sanders2019}'s work which avoids almost all arithmetic during the preparation stage.
    
    In this work, we construct an optimized black box state loading scheme with which various important sets of coefficients can be loaded significantly faster than in $\BigO(\sqrt N)$ rounds of amplitude amplification---up to only $\BigO(1)$ many.
    We achieve this with two variants of our algorithm. The first employs a modification of the oracle from \citeauthor{Sanders2019}, which requires fewer ancillas ($\log_2 g$ vs $g+2$ in the bit precision $g$), and fewer non-Clifford operations per amplitude amplification round within the context of our algorithm. The second utilizes the same oracle, but at slightly increased cost in terms of ancillas ($g+\log_2g$) and non-Clifford operations per amplification round. As the number of amplitude amplification rounds enters as multiplicative factor, our black box state loading scheme yields an up to exponential speedup as compared to prior methods. This speedup translates beyond the black box case.
\end{abstract}

\section{Introduction}
The ability to prepare an arbitrary quantum state is a fundamental building block for many higher-level quantum algorithms, for instance in sparse linear algebra calculations \cite{Harrow2009}, quantum walks \cite{Kempe2003,Santha2008,Berry2009}, machine learning and optimization \cite{Brandao2017,Bravyi2019}, and quantum chemistry or condensed matter simulation \cite{Berry2015,Jones2012}.
%TODO add some more
Starting with \citeauthor{Grover2000} \cite{Grover2000}, a series of  techniques have been developed, based on amplitude amplification \cite{Soklakov2006,Sanders2019,Plesch2011} for the black box regime, or alternative approaches \cite{Mottonen2004,Araujo2020,Grover2002} in case the amplitudes are known \`a priori.
Yet aims and limitations considered in these two regimes differ: generic black box algorithms allow loading amplitudes that stem from oracle subroutines, and others only prepare states for a particular set of amplitudes that has to be known from the outset. Furthermore, precision, requirements in terms of ancillas or intermediate measurements, or restrictions that have to be placed on the amplitudes to be loaded are not standardised, and vary from case to case.

In brief, quantum state preparation addresses the task to transform an amplitude vector $\alpha = (\alpha_0, \ldots, \alpha_{N-1})$ into a quantum state close to
\[
    \frac{1}{\| \alpha \|_2} \sum_{i=0}^{N-1} \alpha_i \ket i
    \quad\text{or}\quad
    \frac{1}{\sqrt{\| \alpha \|_1}} \sum_{i=0}^{N-1} \sqrt{\alpha_i} \ket i.
\]
The first ``linear coefficient'' problem, as dubbed by \cite{Sanders2019}, is the one we consider in this work as the generic quantum state loading task.
In this form, it is a crucial component e.g.~for quantum linear algebra routines (HHL, to prepare an initial state $\ket b$ for which to solve $A\ket x = \ket b$), or random walks (to prepare an initial state whose amplitudes can be given by an arithmetic formula such as $\alpha_x \sim \sin(\pi (x+1)/(N+1))$ for a vertex $x$ of a graph, \cite{Childs2009}).

If the amplitudes $\alpha_i$ are all known beforehand, circuits that hard-code the given data into gates can be developed \cite{Mottonen2004,Araujo2020}. 
As shown in \cite{Mottonen2004}, to load an $N$-amplitude state requires a circuit depth of $\Omega(N)$.
In \cite{Araujo2020}, the authors trade this depth (which is exponential in the number of qubits) for $N$ ancillas, and an $\Theta(\log^2N)$ depth; yet this tradeoff has to be taken with a grain of salt, as long-range gates, when broken down to an e.g.~linear (or grid-like) qubit topology require $\sim N$ swaps to be performed throughout, and the resulting state is left entangled with the ancillas, i.e.\ a superposition of states $\ket i\ket{\psi_i}$ instead of just basis states $\ket i$.
Furthermore, learning-based methods exist that yield circuits which approximately produce the target state, for instance based on generative adversarial neural networks \cite{Zoufal2019}. Since the ansatz circuit can be  chosen freely, depending on the required accuracy ciruit depths less than $\BigO(N)$ and without additional ancilla requirements can be achieved.

Instead of the case where the amplitude vector $\alpha$ is known from the outset, the focus in this paper is on black box algorithms based on oracles access to $\alpha$: as is common, the procedure for creating such a state is based on the existence of an oracle $\op U_\alpha$, which acts on a bipartite state
\begin{equation}\label{eq:Ua}
 \op U_\alpha \ket i\ket j = \ket i \ket{j \oplus \bitalpha_i},
\end{equation}
where we assume $\bitalpha_i$ is a $g$-bit approximation to $\alpha_i$, and $\ket j$ is a $g$-bit register.
With the aid of amplitude amplification, this oracle is then raised to an overall procedure to produce a state with amplitudes close to $\alpha$.
To date, the best-known asymptotic runtime (measured in the number of necessary oracle calls) is $\sim \sqrt{N}$ \cite{Grover2000,Soklakov2006,Sanders2019,Plesch2011}.

\paragraph{Our Contribution.}
In this paper, we improve upon the state-of-the art, i.e.~\citeauthor{Sanders2019}'s black box algorithm, by significantly improving loading times for black box quantum states if some prior knowledge about the amplitudes is known. This prior information is ``cheaply'' available from the oracle, in the form of the amplitudes' \emph{average bit weight}: if the amplitudes are given as a sequence of binary numbers, the expected value of the $j$\textsuperscript{th} bit is simply the average over all $j$\textsuperscript{th} bits in the sequence. An approximation to these expected bit values can be measured quickly by a few oracle calls and one-qubit measurements.
We prove that for various distributions this additional piece of information can significantly reduce the necessary number of amplitude amplification rounds, from $\BigO(\sqrt N)$ to up to $\BigO(1)$, i.e.\ an up to \emph{exponential speedup} for state preparation, and in various settings relevant for real-world scenarios (cf.~\cref{tab:speedups,sec:explicit}).

We develop our black box state loading algorithm based on a variant of the \cite{Sanders2019,Grover2002} oracle, which, within the context of our algorihtm, is significantly more frugal in terms of ancillas required for the same precision to be loaded: for instance for amplitudes with 32 bits of precision, \cite{Sanders2019} require at least 34 additional ancillas, whereas for our algorithm 5 suffice---at the cost of requiring a slightly deeper circuit.\footnote{We  did not take any device topology into account in this comparison, where more ancillas often mean more reshuffling of information, with a correspondingly higher resulting circuit depth.}
As the fact that we utilize two different oracles means we are not comparing like to like, we also construct a variant of our algorithm based on \emph{the same} oracle as in \cite{Sanders2019}. While this comes at a slight overhead in terms of ancillas and non-Clifford gates required (logarithmic in the precision of the amplitudes to be loaded, to be precise), the asymptotic speedups in terms of number of amplitude amplification rounds prevail and often outstrip the additional resource requirements necessary per round. This means that even in case of utilizing the same oracle as \cite{Sanders2019} an exponential speedup is achievable, as compared to prior methods.

Notably, our proposed improvement goes \emph{beyond} the black box setting.
If the oracle is, for instance, an arithmetic circuit producing the weights $\alpha_i \sim \sin(\pi(i+1)/(N+1))$ for $i=0,\ldots,N$ for a discrete quantum walk application, those weights might well be known beforehand, so pre-compiled state preparation circuits can be applied; but instead of a resulting depth $\Omega(N)$ circuit, already existing black-box algorithms reduce the asymptotics to $\BigO(\sqrt N)$; our algorithm further reduces the runtime to $\BigO(1)$, translating the potential exponential speedup even to the case when amplitudes are known beforehand.

Finally, we generalise our state loading procedure as a generic subroutine able to prepare states like
\[
    \sum_{i=0}^{N-1} \left( \sum_{j=1}^g b_{ij} w_{\!j} \right) \ket i
\]
for some boolean matrix $(b_{ij})_{1 \le i, j \le g}$ and weight vector $(w_{\!j})_{1\le j\le g}$, potentially useful as a linear algebra subroutine, and where the same optimisations for given prior knowledge about the $b_{ij}$ mentioned in the last paragraph can be applied.

\section{Gradient-State Based Black-Box Quantum State Preparation}
The algorithm we propose shares similarities with \cite{Sanders2019}'s method, in the sense that it is a two-step protocol based on preparation of an intermediate state which maps the oracle's amplitudes into an ancillary system, and a second step that collates the ancillary amplitudes into the final form. In either step, amplitude amplification is employed to transform one state to the other.

Yet while \citeauthor{Sanders2019} use a large ancillary register where the $2^g$ dimensions represent $g$ bit numbers on a linear scale $0, 2^{-g}, 2\times2^{-g}, 3\times2^{-g}, \ldots, (2^g-1)\times2^{-g}$, we identify the ancillary dimensions with a logarithmic scale $1/2,1/4,1/8,\ldots,2^{-g}$.
This mapping has two advantages: i.~it requires fewer qubits ($\lceil \log_2 g \rceil$ instead of $g$ many), and ii.~it allows knowledge of the oracle's average bit weight to be exploited, often reducing the necessary number of amplitude amplification rounds below \citeauthor{Sanders2019} and \citeauthor{Grover2000}'s $\BigO(\sqrt N)$.

In \cref{sec:algorithm} we present the state loading protocol in detail, explaining the two stages, and deriving explicit runtime bounds.
\Cref{sec:error} then contains a discussion of error bounds and success probabilities, and in \cref{sec:resources} we describe how the new ingredients that our proposal is based on---preparing the logarithmic scale in the ancilla state, a reduction to the \cite{Sanders2019,Grover2002} oracle, and computing a swap network that replaces \citeauthor{Sanders2019}'s comparator subroutine---can be constructed using few additional resources.
\Cref{sec:bootstrapping} then discusses how some initial calls to the oracle can be used to produce an optimised initial state that results in a reduced number of necessary amplitude amplification rounds.
We further present comparisons for number of ancillas and non-Clifford resources in \cref{tab:precision}, and speedups obtained by using the optimised bootstrapping technique in \cref{tab:speedups}.

\subsection{The Algorithm}\label{sec:algorithm}
Throughout the paper, $\| x \|_p \coloneqq (\sum_i |x_i|^p)^{1/p}$ denotes the standard $\ell_p$ norm.
Given a discrete $N$-element list of nonnegative amplitudes $\alpha = (\alpha_0, \ldots, \alpha_{N-1})$ with $\| \alpha \|_2 = 1$,
we let $2^g \bitalpha_i \in \field N$ be a binary approximation to $\alpha_i$ to $g \in \field N$ bits of precision (rounded towards zero), and let analogously $\bitalpha=(A_0, A_1, \ldots, A_{N-1})$ be a $g$-bit approximation of the entire amplitude vector.
The $j$\textsuperscript{th} bit of $\bitalpha_i$ is then denoted $\bitalpha_{ij}$, in Little Endian order, i.e.\ such that $\bitalpha_{i0}$ is the most significant bit (denoting the $1/2$'s); with this we simply have $\bitalpha_i = \sum_j 2^{-j} \bitalpha_{ij}$.
If not stated differently we assume that $N=2^n$ for some $n\in\field N$; but the construction is generalised readily to non-power-of-2 numbers (e.g.\ by padding the vector $\alpha$ suitably by zeroes).
Our goal is to prepare the state
\begin{equation}\label{eq:final-state}
    \ket{\bitalpha} = \frac{1}{\| \bitalpha \|_2} \sum_{i=0}^{N-1} \bitalpha_i \ket i.
\end{equation}
How far does this state deviate from $\ket\alpha = \sum_{i=0}^{N-1} \alpha_i \ket i$, in trace distance $\sqrt{1-|\braket{\alpha}{\bitalpha}|^2}$?
We first note
\[
    | \braket{\alpha}{\bitalpha} | = \frac{1}{\| \bitalpha \|_2} \sum_{i=0}^{N-1} \alpha_i \bitalpha_i \overset*\ge \frac{1}{\| \bitalpha \|_2} \sum_{i=0}^{N-1} \bitalpha_i \bitalpha_i = \| \bitalpha \|_2 \ge \sqrt{1- 2\times 2^{-g} \| \alpha \|_1},
\]
where in the step marked with $(*)$ we have made use of the fact that the $\bitalpha_i$ is rounded towards zero.
Then
\begin{equation}\label{eq:alpha-g-tracedist}
\sqrt{1-|\braket{\alpha}{\bitalpha}|^2} \le 2^{(1-g)/2} \sqrt{\| \alpha \|_1}. 
\end{equation}

In order to obtain a copy of this state $\ket{\bitalpha}$ with high likelihood, we first define the amplitude gradient state
\begin{equation}\label{eq:gradient-state}
    \ket{g}_G \coloneqq \frac{1}{\sqrt{2^g-1}} \sum_{j=0}^{g-1} 2^{(g-j-1)/2} \ket j,
\end{equation}
such that e.g.
\[
    \ket{1}_G = \ket 0, \ \ \ \ket{2}_G = \frac{1}{\sqrt 3} \left( \sqrt2 \ket0 + \ket1 \right), \ \ \  
    \ket{3}_G = \frac{1}{\sqrt 7} \left( \sqrt4 \ket0 + \sqrt2 \ket1 + \ket2 \right), \ \ \ \text{etc.}
\]
We assume for now that this state can be prepared from a qudit ancilla $\ket0 \in \field C^g$ with a suitable unitary operation $\op G$ (for implementation details see \cref{sec:g-state}), where we can again for simplicity assume that $g$ is a power of 2.
A combination of Hadamard operations and $\op G$ then allows us to prepare the initial state
\begin{equation}\label{eq:initial-state}
    \ket s \coloneqq  \left[ \op H^{\otimes n}\otimes \op G \right] \ket0^{\otimes n} \ket0 = \frac{1}{\sqrt N} \sum_{i=0}^{N-1} \ket i \ket g_G.
\end{equation}

The Grover oracle for this setup is given by a unitary operation, defined on basis states $\ket i\ket j \in (\field C^2)^{\otimes n} \otimes \field C^g$ via
\begin{equation}\label{eq:Uw}
    \op U_\omega \ket i \ket j = (-1)^{\bitalpha_{ij}} \ket i \ket j.
\end{equation}
In other words, $\op U_\omega$ can compute bits of the coefficient $\alpha_i$, and flips the sign indicating the $j$\textsuperscript{th} bit of the $i$\textsuperscript{th} index if and only if $\bitalpha_{ij} = 1$.

For instance on the initial state $\ket s$, $\op U_\omega$ takes the action
\[
    \op U_\omega \ket s = \frac{1}{\sqrt N} \sum_{i=0}^{N-1} \ket i \times \frac{1}{\sqrt{2^{g}-1}} \sum_{j=0}^{g-1} 2^{(g-j-1)/2} (-1)^{\bitalpha_{ij}} \ket j.
\]
There is various ways such an oracle unitary with phase kickback can be constructed; the most efficient means will depend on the target distribution $\alpha$; for now we simply assume its existence.
We emphasize that $\op U_\omega$ is naturally not the same oracle as in \cite{Grover2002,Sanders2019} that can write out an entire amplitude $\op U_\mathrm{amp}\ket i\ket z=\ket i\ket{z \oplus A_j}$ in one query.
While the black box state preparation algorithm we present in this paper is fundamentally based in $\op U_\omega$, we discuss in \cref{sec:Uw} how to proceed from $\op U_\mathrm{amp}$ as a starting point, which will incur both a qubit and gate overhead. In terms of query complexity to $\op U_\omega$ or $\op U_\mathrm{amp}$, however, the two settings will prove to be identical.

The final ingredient is a variant of the Grover diffusion operator, namely
\begin{equation}\label{eq:Us}
    \op U_s \coloneqq \left[ \op H^{\otimes n} \otimes \op G \right] \left[ 2\ketbra 0 - \1 \right] \left[ \op H^{\otimes n} \otimes \op G \right]^\dagger.
\end{equation}

Defining the intermediate target state
\begin{equation}\label{eq:omega}
    \ket\omega \coloneqq \frac{1}{\sqrt{ \| \bitalpha \|_1 }} \sum_{i=0}^{N-1} \ket i \sum_{j=0}^{g-1} 2^{-(j+1)/2} \bitalpha_{ij} \ket j
\end{equation}
we can---as a well-known step---equate both $\op U_s$ and $\op U_\omega$ with two Householder operations.
\begin{lemma}\label{lem:householder}
$\op U_s = 2 \ketbra s - \1$, and $\op U_\omega = \1 - 2 \ketbra\omega$ when restricted to the subspace spanned by $\ket s$ and $\ket\omega$.
\end{lemma}
\begin{proof}
By \cref{eq:Us}, we have that
\[
    \op U_s = 2 \left( \left[\op H^{\otimes n} \otimes \op G\right] \ketbra 0 \left[\op H^{\otimes n} \otimes \op G\right]^\dagger\right) - \1
    =2 \ketbra s - \1
\]
by definition of the initial state $\ket s$ in \cref{eq:initial-state}.
% NEW IDEA
To prove that $\op U_\omega$ is also a Householder transformation within the subspace spanned by the initial state $\ket s$ (\cref{eq:initial-state}) and the intermediate target state $\ket\omega$ (\cref{eq:omega}), we proceed by direct calculation.
\begin{align}
    \op U_\omega \ket s
    &= \frac{1}{\sqrt N} \sum_{i=0}^{N-1} \ket i \times \frac{1}{\sqrt{2^{g}-1}} \sum_{j=0}^{g-1} 2^{(g-j-1)/2} (-1)^{\bitalpha_{ij}} \ket j \nonumber\\
    &\overset*= \ket s + \frac{1}{\sqrt N} \sum_{i=0}^{N-1}\ket i \times
    \frac{1}{\sqrt{2^g-1}} \sum_{j=0}^{g-1} \left[
    2^{(g-j-1)/2} (-1)^{\bitalpha_{ij}} \ket j - 2^{(g-j-1)/2} \ket j
     \right] \nonumber\\
    &= \ket s - \frac{2}{\sqrt N} \frac{1}{\sqrt{2^g-1}} \sum_{i=0}^{N-1} \ket i \sum_{j=0}^{g-1}2^{(g-j-1)/2} \bitalpha_{ij} \ket j \nonumber\\ 
    &= \ket s - 2\frac{\sqrt{\| \bitalpha \|_1}}{\sqrt{N}}\frac{2^{g/2}}{\sqrt{2^g-1}} \ket\omega, \label{eq:U_w_on_s}
\end{align}
where in the step marked with $*$ we have added and subtracted $\ket s$; the term on the right hand side of the square brackets in that line stems from the summand in the gradient state, \cref{eq:gradient-state}.
As 
\begin{align}
    \braket{\omega}{s} &= \frac{1}{\sqrt N\sqrt{\| \bitalpha\|_1}} \sum_{i=0}^{N-1}\sum_{i'=0}^{N-1} \braket{i}{i'} \frac{1}{\sqrt{2^g-1}} \sum_{j=0}^{g-1}\sum_{j'=0}^{g-1}2^{(g-j'-1)/2}2^{-(j+1)/2}\bitalpha_{ij}\braket{j}{j'} \nonumber\\
    &= \frac{1}{\sqrt N\sqrt{\|\bitalpha\|_1}}\frac{2^{g/2}}{\sqrt{2^g-1}} \sum_{i=0}^{N-1}\underbrace{\sum_{j=0}^{g-1}2^{-j-1}\bitalpha_{ij}}_{\equiv \bitalpha_i} \nonumber\\[-4mm]
    &=\frac{\sqrt{\| \bitalpha\|_1}}{\sqrt N} \frac{2^{g/2}}{\sqrt{2^g-1}}, \label{eq:<w|s>}
\end{align}
we see that \cref{eq:U_w_on_s} further simplifies to $\op U_\omega\ket s = \ket s - 2\braket{\omega}{s}\ket w$.

The last operation to verify is
\begin{align*}
    \op U_\omega \ket\omega = \frac{1}{\sqrt{\|\bitalpha\|_1}} \sum_{i=0}^{N-1} \ket i \sum_{j=0}^{g-1} 2^{-(j+1)/2} \bitalpha_{ij} (-1)^{\bitalpha_{ij}} \ket j
    = -\ket\omega,
\end{align*}
as $\bitalpha_{ij}(-1)^{\bitalpha_{ij}} = -\bitalpha_{ij}$.
\end{proof}

In combination, \cref{lem:householder} thus forms the Grover iteration operator $\op U_s \op U_\omega$, which performs a rotation in the space spanned by $\ket s$ and $\ket\omega$.

We observe that the fact that \cref{lem:householder} only shows that $\op U_\omega$ has Householder form only when restricted to this hyperplane is also the case for normal Grover search, e.g.\ in the case when there are multiple marked target elements. Apart from the explicit proof we just gave for \cref{lem:householder}, it is easy to see that this must be the case, as the oracle $\op U_\omega$ given in  \cref{eq:Uw} has precisely the same form as the standard Grover oracle; the explicit form of the two vectors $\ket s$ and $\ket \omega$ thus emerge from the Grover diffusion operator $\op U_s$, which implicitly defines the initial state $\ket s$.

Starting from this initial state $\ket s$, we can thus use these Grover iterations to approach the intermediate target state $\ket\omega$ step-by-step. Instead of deriving the runtime of this variant of Grover search (which is by now standard and can e.g.\ be found in \cite[Sec.6.1.3]{Nielsen_and_Chuang}), we rely on the state-of-the-art in fixed-point amplitude amplification in order to not overshoot the target \cite{Yoder2014}.
We cite their result here in a concise form for completeness.
\begin{lemma}[Fixed-Point Amplitude Amplification, \cite{Yoder2014}]
	Let $n\in\field N$. We are given a unitary operator $\op S$ that prepares an initial state $\ket s = \op S\ket 0^{\otimes n}$, and a target state $\ket T$ with $\braket{T}{s}=\sqrt\lambda \ee^{\ii\xi}$, $\lambda,\xi\in\field R$, and an oracle $\op U\ket T\ket b=\ket T\ket{b\oplus 1}$ and $\op U\ket{\bar T}\ket b=\ket{\bar T}\ket b$ for $\braket{\bar T}{T}=0$.
	Then there exists a fixed point amplitude amplification procedure that, when starting from the initial state $\ket s$, extracts the target state $\ket T$ with success probability $\ge 1 - \delta^2$ for some $\delta\in[0, 1]$, utilizing $L \sim \log(2/\delta) / \sqrt \lambda$ calls to the oracle $\op U$ interleaved by $2L$ calls to $\op S$.
\end{lemma}
A success probability of $1-\delta^2$ means that after the given number of rounds the fixed-point amplification procedure results in a state $\ket{T'}$ which satisfies
\[
1-\delta^2 = \operatorname{Tr}(\ketbra T\ketbra{T'}) = |\braket{T}{T'}|^2.
\]

For us, $\op S = \op H^{\otimes n} \otimes \op G$, as is clear from \cref{eq:Us,eq:initial-state}, and the oracle $\op U=\op U_\omega$ is the same with a phase kickback mechanism implied (as already explained).
Further translating this terminology to our setting, we start with the initial state $\ket s$, and amplify to an intermediate target state $\ket T=\ket\omega$; where then $\ket{\omega'}$ denotes the state we \emph{actually} reach after the amplitude amplification has been applied.

Denoting the number of fixed point amplification rounds with $L_1$ (where the subscript indicates that our overall algorithm will have a secondary stage with similar quantities subscripted by $2$, to be determined in due course), we need
\begin{equation}\label{eq:fp-aa}
    L_1 \sim \frac{\log(2/\delta_1)}{\sqrt \lambda_1}
    \quad\text{where}\quad
    \delta_1 = \sqrt{1-|\braket{\omega}{\omega'}|^2}
    \quad\text{and}\quad
    \lambda_1 = | \braket{\omega}{s} |^2.
\end{equation}
The parameter $\delta_1$ is thus a guaranteed trace distance to the intermediate target state that we will reach; and $\lambda_1$ indicates the overlap of the target coefficient vector $\alpha$ with the initial state $\ket s$.

Then with \cref{eq:<w|s>,eq:alpha-g-tracedist} we have
\begin{align}
    \sqrt{\lambda_1} = \braket{\omega}{s} 
    &= \sqrt{\frac{\| \bitalpha \|_1}{N}} \sqrt{\frac{2^g}{2^g-1}} \ge \sqrt{\frac{\| \alpha \|_1 - 2^{-g}N}{N}}.  \label{eq:lambda1}
\end{align}
The two extremes here are the cases where $\alpha$ itself is already uniform---such that $\alpha_i = 1/\sqrt N$, which yields $L_1 = \BigO(N^{1/4})$; and a delta distribution---which yields $L_1 = \BigO(N^{1/2})$.
The reason why the uniform example does not already give $L_1 = \BigO(1)$ is that $\ket s$ does \emph{not} represent a uniform intermediate target state---in fact, $\ket s$ just is a uniform superposition of all possible bits, whereas a uniform superposition will likely have $\bitalpha_{ij} = 0$ for a series of higher-significance positions.
We improve upon this caveat in \cref{sec:bootstrapping}.

How do we move from the intermediate target state $\ket\omega$ from \cref{eq:omega} to the target state $\ket{\bitalpha}$?
We must project the second register onto the gradient state; either by postselection, or by another nested round of amplitude amplification.
We note that the fidelity of amplification in this step only increments the probability of obtaining the state $\ket{\bitalpha}$; if we succeed, we are guaranteed that we are left with that exact state, and there is no more introduced error.

If we perform amplitude amplification, we denote with $L_2,\delta_2$ and $\lambda_2$ the associated amplitude amplification parameters from \cref{eq:fp-aa}.
If we define the projector $\Pi \coloneqq \1 \otimes \ketbra g_G$, we can calculate
\begin{align}
    \lambda_2 = \bra\omega \Pi \ket\omega 
    &= \frac{1}{\| \bitalpha \|_1} \frac{1}{2^g-1} \sum_{i=0}^{N-1} \left(
    \sum_{j=1}^g 2^{(g-j)/2} 2^{-j/2} \bitalpha_{ij} \right)^2  \nonumber\\
    &= \frac{1}{\| \bitalpha \|_1} \frac{2^g}{2^g-1} \sum_{i=0}^{N-1} (\bitalpha_i)^2
    \ge \frac{\| \bitalpha \|_2^2}{\| \bitalpha \|_1}. \label{eq:lambda_2-eq1}
\end{align}
As $\bitalpha_i \le \alpha_i$ (since we round towards zero), we have that
\[
\frac{\| \bitalpha \|_2^2}{\| \bitalpha \|_1} 
\ge \frac{ \| \alpha - 2^{-g} \|_2^2 }{ \| \alpha \|_1 } = \frac{1}{\| \alpha \|_1} - 2\times 2^{-g} +  \frac{2^{-2g}N}{\| \alpha \|_1} \ge \frac{1}{\| \alpha \|_1} - 2\times 2^{-g}.
\]
%NOTE the above calculation is best performed coefficient-wise; and the notation is a bit misleading, one should first part out all \alpha_i that are zero. But this is only slightly abusing the notation.
With $\bra\omega \Pi \ket\omega \approx 1 / \| \alpha \|_1$, as before, we observe that the 1-norm of the coefficients $\alpha$ determines the overlap, with extreme cases being the uniform distribution for which $\| \alpha \|_1 = \sqrt N$, and delta distribution for which $\| \alpha \|_1 = 1$.

Jointly together these two overlaps yield the overall runtime: we have
\begin{align*}
    \lambda_1\lambda_2 &\ge \frac{\| \alpha \|_1 - 2^{-g}N}{N}\left( \frac{1}{\| \alpha \|_1} - 2\times 2^{-g} \right)
    = \frac1N - 2^{-g}\left( \frac{1}{\| \alpha \|_1} -2 \frac{\| \alpha \|_1}{N} - 2\times 2^{-g} \right) \\
    &\ge \frac1N - \frac{2^{-g}}{\| \alpha \|_1}.
\end{align*}
As the two amplitude amplification subroutines have to be called in a nested fashion, we have
\begin{align*}
 L = L_1 L_2 &= \frac{\log(2/\delta_1)\log(2/\delta_2)}{\sqrt{\lambda_1\lambda_2}} \le \log(2/\delta_1)\log(2/\delta_2) \bigg/\sqrt{ \frac1N - \frac{2^{-g}}{\| \alpha \|_1} } \\
 &\overset{*}{\le}
 \log(2/\delta_1)\log(2/\delta_2)\left( \sqrt N + \frac{2^{-g}N}{\| \alpha \|_1} \right),
\end{align*}
where the last inequality marked with $(*)$ holds for high enough precision $g$ such that $x = 2^{-g}N / \| \alpha \|_1 < 1/2$, and using $(1-x)^{-1} \le 1 + x$ for $x \le 1/2$.\footnote{
We note that this bound can be made tighter by a constant factor up to two, by expanding $(1-x)^{-1} \le 1+yx$ for some $y\in(1/2, 1]$ and correspondingly demanding $x < (y + \sqrt{y(4+y)} - 2)/(2y)$.
}
Under this assumption, we have
\begin{equation}\label{eq:scaling}
 L = L_1L_2 \le \log(2/\delta_1)\log(2/\delta_2) \left(\sqrt N + \frac12\right).
\end{equation}

\subsection{Precision and Error Analysis}\label{sec:error}

\begin{table}[t]
	\hspace{2mm}\begin{tabular}{rrllllll}
		\toprule
		& bit precision $g$ & 2 & 4 & 8 & 16 & 32 ($>$float) & 64 ($>$double) \\
		&$\epsilon=2^{-g}\le$ & $\displaystyle\frac14$ & $\displaystyle\frac{1}{16}$ & $0.004$ & $1.6 \times 10^{-5}$ & $ 2.3 \times 10^{-10}$ & $5.5 \times 10^{-20}$ \\
		\midrule
		\multirow{4}{*}{Toffolis} & \cite{Sanders2019} variant 1 & 4 & 8 & 16 & 32 & 64 & 128 \\
		& \cite{Sanders2019} variant 2 & 6 & 14 & 30 & 62 & 126 & 254 \\
		& us, variant 1 & \multicolumn{6}{l}{none} \\
		& us, variant 2 & 3 & 12 & 33 & 78 & 171 & 360 \\[5mm]
		\multirow{2}{*}{$\sqrt{\mathrm{SWAP}}$} & \cite{Sanders2019} & \multicolumn{6}{l}{none} \\
		& us & 2 & 4 & 8 & 16 & 32 & 64 \\
		\midrule
		\multirow{4}{*}{ancillas} & \cite{Sanders2019} variant 1 & 5 & 9 & 17 & 33 & 65 & 129 \\
		& \cite{Sanders2019} variant 2 & 4 & 6 & 10 & 18 & 34 & 66 \\
		& us, variant 1 & 1 & 2 & 3 & 4 & 5 & 6 \\
		& us, variant 2 & 3 & 6 & 9 & 20 & 37 & 70 \\
		\bottomrule
	\end{tabular}
	\caption{Reachable precision for a given gradient state $\ket g_G$ on $q$ qubits, and non-Clifford gate counts for each amplitude amplification iteration.
	\cite{Sanders2019}'s variant 1 and 2 refer to the two options for implementing their comparator circuits ($2g+1$ ancillas, $2g$ Toffolis; resp.~$g+2$ ancillas, $4g-2$ Toffolis). Our two variants refer to utilizing $\op U_\omega$ or $\op U_\mathrm{amp}$ as oracle, as discussed in \cref{sec:comparator} ($\log_2 g$ ancillas, no Toffolis; resp.~$g+\log_2g$ ancillas, $\le2g\log_2 g$ Toffolis; and $g$ $\sqrt{\mathrm{SWAP}}$s in either case for the gradient state). 
	Our SWAP network for variant 2 is built as in \cref{fig:permutation}, and trivial gate optimizations are applied.
	As \cite{Sanders2019} use $\op U_\mathrm{amp}$, only our variant 2 should be directly compared to \cite{Sanders2019}. We refer the reader to \cref{sec:bootstrapping,tab:speedups} for a speedup with regards to amplitude amplification rounds necessary.
	Device topology is not taken into account, where for planar architectures it is often the case that more ancillas mean more reshuffling of information, with correspondingly higher circuit depth.}
	\label{tab:precision}
\end{table}

For runtime estimates we already kept track of lower bounds as given in \cref{eq:scaling}.
Furthermore, we can independently tune the success probability $\delta_2$ which determines with what likelihood we obtain a copy of the state $\ket{\bitalpha}$ (or, more precisely, the state $\ket{A'}$ that we obtain when starting from an approximate intermediate target state $\ket{\omega'}$ instead of $\ket\omega$ in the previous amplification round), and depending on what failure rate we find acceptable; in other words, $\delta_2$ determines with what likelihood we obtain a state $\ket\psi$, $\delta_1$ controls how close $\ket\psi$ to $\ket{\bitalpha}$ a state we obtain, and the parameter $g$ controls the precision of the coefficients (i.e.\ how close $\ket{\bitalpha}$ is to $\ket\alpha$).

As our overall procedure approximates $\ket\alpha$ by $\ket{\bitalpha}$, which by \cref{eq:alpha-g-tracedist} differ in trace distance by $\le 2^{(1-g)/2} \sqrt{\| \alpha \|_1}$, we should choose $\delta_1$ to represent a trace distance of about the same amount (resulting in $\sqrt2$ times this distance overall).
Then \cref{eq:scaling} reads
\[
    L \le \log\left(\frac{2}{\delta_2}\right)\times\frac12\big( 1 + g - \log \| \alpha \|_1 \big) \times \left( \sqrt N + \frac12 \right).
\]
This is in alignment with \cite{,Grover2000,Sanders2019}, who quote an approximate runtime bound for their ``linear coefficients'' problem of $\BigO(g \sqrt N)$, where the authors omitted the logarithmic error terms due to $\delta_1$ and $\delta_2$.
One can verify that keeping track of amplitude amplification errors throughout their procedure yields equivalent factors.

As the gradient state can store $g$ values within $\lceil \log_2 g \rceil$ many qubits, we can load high-precision numbers with little overhead, as shown in \cref{tab:precision}.

\subsection{Computational Primitives and Resource Requirements}\label{sec:resources}
There are two fundamental building blocks that we need to analyse: implementing an amplitude gradient gate for states like \cref{eq:gradient-state}, and the phase kickback oracle unitary $\op U_\omega$ from \cref{eq:Uw} in case we only have access to an oracle as in \cite{Sanders2019,Grover2002}.

\subsubsection{Oracle Comparison between $\op U_\omega$ and $\op U_\mathrm{amp}$}\label{sec:Uw}
\citeauthor{Grover2002,Sanders2019} analyse the complexity of their black box state preparation algorithms by assuming there exists an oracle unitary $\op U_\mathrm{amp}$ that can calculate the desired amplitudes digitally, as
\begin{equation}\label{eq:Uamp}
    \op U_\mathrm{amp} \ket i \ket z = \ket i \ket{z \oplus \bitalpha_i}
\end{equation}
for any bit string $z \in \{0,1\}^g$ and computational basis state $\ket i$.
How this oracle unitary is implemented is left unspecified; in a sense this notion completely decouples the aspect of \emph{computing} the coefficients from \emph{loading} these coefficients as actual amplitudes.
We will keep with this notion, as the oracle's own resource requirements are highly dependent on the amplitudes to be loaded. For similar reasons, a direct comparison of the resource requirements between $\op U_\omega$ and $\op U_\mathrm{amp}$ is futile.

In addition, comparing two black box state loading algorithms that depend on different oracles is only fair if neither of the oracles yields an unfair advantage over the other. To rule this out, we will show in the following discussion that we can phrase a variant of our black box state loading procedure that uses $\op U_\mathrm{amp}$ as an oracle instead of $\op U_\omega$, which will put both algorithms on a level playing field.
This will induce a minor overhead in terms of ancillas, but crucially results in the same query complexity, as exactly one application of $\op U_\mathrm{amp}$ will be required to replace one application of $\op U_\omega$.

In conjunction with the oracle, \citeauthor{Sanders2019} employ a non-Clifford step that requires a comparator circuit mediated by a conditional flip of a flag qubit
\begin{equation}\label{eq:comparator}
    \ket{A}\ket{B}\ket{0} \longmapsto \ket{A}\ket{B}\begin{cases}
    \ket 1 & A \ge B \\ \ket 0 & A < B
    \end{cases}
\end{equation}
which is used to transduce the amplitudes into an ancillary system (cf.~\cite[eq.~6]{Sanders2019}).

In contrast, the algorithm we introduced in this paper assumes the existence of an oracle $\op U_\omega$ (\cref{eq:Uw}), which induces a conditional phase flip $\op U_\omega\ket i\ket j = (-1)^{\bitalpha_{ij}}\ket i \ket j$.
When applied to a gradient state in place of register $\ket j$, one application of $\op U_\omega$ thus achieves that a total weight of coefficients proportional to $\bitalpha_i$ has a flipped coefficient, which is equivalent to the combined step in \citeauthor{Sanders2019}. Again, it is important to emphasize that this does not mean the oracles are equivalent, or that query complexity for the two routines should be comparable a priori.

However, in \cref{sec:comparator} we will demonstrate directly that with the aid of a comparator circuit one can utilize $\op U_\mathrm{amp}$ to implement $\op U_\omega$; and thus obtain a variant of our proposed algorithm based on $\op U_\mathrm{amp}$ as fundamental oracle primitive.
As this renders the algorithms more comparable, we provide a resource calculation for either choice of oracle in \cref{tab:speedups}.

\subsubsection{Comparator Circuits, and Replacing $\op U_\omega$ with $\op U_\mathrm{amp}$}\label{sec:comparator}
In \citeauthor{Sanders2019}'s comparator circuit, comparing two $g$ bit numbers as in \cref{eq:comparator} can be done using $g$ Toffoli gates \cite{Cuccaro2004,Gidney2018}.
As the method for this comparison relies on binary addition (or rather subtraction) which in this case is performed in-place, the register $\ket{\bitalpha_i}$ first has to be copied to a temporary slot (which is always possible as it is a computational basis state).
The comparator itself thus requires $g+1$ ancillas, and $2g$ Toffoli gates (as the comparison subroutine has to be uncomputed as well).

In our case, if we amend our algorithm to be based on $\op U_\mathrm{amp}$ instead of $\op U_\omega$, we first apply
\[
	\op U_\mathrm{amp}\frac{1}{\sqrt N}\sum_{i=0}^{N-1}\ket i\ket 0\ket g_G = \frac{1}{\sqrt N}\frac{1}{\sqrt{2^g-1}}\sum_{i=0}^{N-1}\ket i\ket{A_i}\sum_{j=0}^{g-1}2^{(g-j-1)/2}\ket j.
\]
As the register $\ket j$ stems not from uniformly-distributed bins, but from a binary tree partition (i.e.\ from the gradient state $\ket g_G$) the next step is to introduce a phase flip on $\ket j$ if the $j$\textsuperscript{th} bit of $\bitalpha_i$ is $1$, i.e.\ when $\bitalpha_{ij} = 1$.

As aforementioned, with the use of $g$ ancillas it is straightforward to translate $\ket j$ into a unary one-hot mask and then perform such a bit-by-bit comparison with $2g$ Toffoli gates, which matches the bound given in \cite{Sanders2019}.

If we want to be qubit-conservative (and assuming $\ket j$ remains stored in binary), one can alternatively decompose a $\lceil \log_2 g \rceil + 1$-control Toffoli gate into primitive resources (to e.g.\ check the address register $\ket j =\ket 5$, and that $\bitalpha_{i5}=1$, and conditionally flip a single ancilla). This can be done with $32\lceil \log_2 g \rceil - 96$ $\op T$ or $\op T^\dagger$ gates \cite{He2017}.
This method results in roughly $8 \log_2 g$ times the Toffoli gate count as in \cite{Sanders2019}.

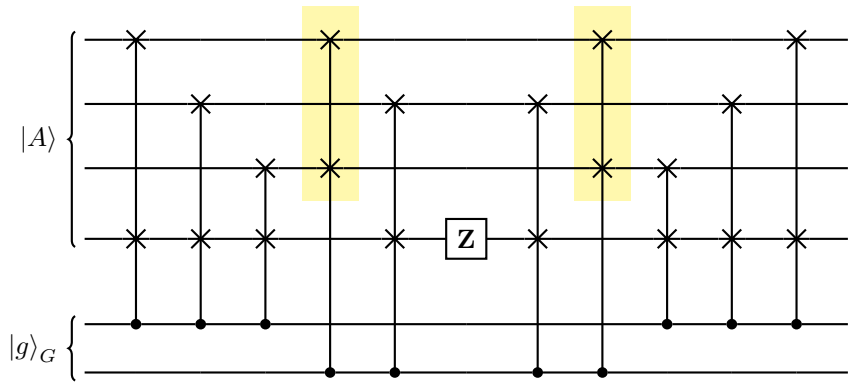
\begin{figure}
	\centering
	\begin{quantikz}
		\lstick[wires=4]{$\ket{\bitalpha}$}\qw & \swap{4} & \qw & \qw & \swap{5}\gategroup[3,steps=1,style={draw=none,fill=yellow!40, inner xsep=2pt},background]{{}}  & \qw & \qw & \qw & \swap{5}\gategroup[3,steps=1,style={draw=none,fill=yellow!40, inner xsep=2pt},background]{{}} & \qw & \qw & \swap{4} & \qw \\
		\qw & \qw & \swap{3} & \qw & \qw & \swap{4} & \qw & \swap{4} & \qw & \qw & \swap{3} & \qw & \qw \\
		\qw & \qw & \qw & \swap{2} & \targX{} & \qw & \qw & \qw & \targX{} & \swap{2} & \qw& \qw& \qw\\
		\qw & \targX{} & \targX{}& \targX{} & \qw &\targX{} & \gate{\op Z}&\targX{}& \qw & \targX{} & \targX{}& \targX{} & \qw\\[3mm]
		\lstick[wires=2]{$\ket{g}_G$}\qw & \control{} & \control{}& \control{} & \qw & \qw & \qw & \qw & \qw & \control{} & \control{}& \control{} & \qw\\
		\qw & \qw & \qw & \qw & \control{} & \control{} & \qw & \control{} & \control{} & \qw & \qw& \qw & \qw
	\end{quantikz}
	\caption{Permutation network for shuffling the $x$\textsuperscript{th} qubit of the register $\ket{\bitalpha}$ to its least significant position, applying a phase flip there, and shuffling the qubits back to their original position.
		Shown is the example $g=2$, which yields $4$ bits of precision for the data to be loaded, as shown in \cref{tab:precision}.
		The network is straightforward to generalize for larger $g$, and requires at most $2g \lceil \log_2 g \rceil$ controlled SWAP operations; in many cases, further optimizations can be applied (in this case the two shaded swaps are redundant because they are outside of the causality cone of the $\op Z$-gate, cf.~\cref{tab:precision}).
		A controlled SWAP operation---or Fredkin gate---can be implemented using a single Toffoli and two CNOTs \cite{Smolin1996}.}
	\label{fig:permutation}
\end{figure}

A better approach is to use a custom permutation network as shown in \cref{fig:permutation}, which shuffles the $j$\textsuperscript{th} bit to the bottom position, where a $\op Z$ operation can be performed (in essence using the register where $\ket A$ is written as a source for phase kickback).
This results in the use of $2 g \log_2 g$ Toffoli gates, i.e.\ an overhead of $\log_2 g$ as compared to \cite{Sanders2019}; the resulting SWAP networks can often be optimized further, by considering the light cone of the phase flip $\op Z$ that is applied to the $j$\textsuperscript{th} bit.
We summarise our resulting Toffoli counts in \cref{tab:precision}.

As a final observation we note that if the oracle $\op U_\omega$ is implemented such that one can query individual bits of the vector $\bitalpha$ to be loaded, no comparison is necessary at all; while this looks like sweeping unwanted complexity of the algorithm under the rug (it is certainly reasonable to assume that in order to calculate the $k$\textsuperscript{th} binary digit of a number one has to also compute all previous $k-1$ bits), it is educational to discuss this variant of the query model, as it demonstrates that the choice of oracle (even if they serve a similar purpose) dictates the query cost of an algorithm. In the same light, if we were to utilize $\op U_\omega$ as a primitive in \cite{Sanders2019}'s work, we would have to query $\op U_\omega$ at least $g$ times in order to obtain a single amplitude written out in binary, what $\op U_\mathrm{amp}$ is capable of achieving in one step.

\subsubsection{Amplitude Gradient State}\label{sec:g-state}
Let us finally turn our attention to the gate complexity of preparing an amplitude gradient state $\ket g_G$ as in \cref{eq:gradient-state}.
Assuming for now that the state space is that of a qudit $\field C^g$, and with access to arbitrary rotations around $\op X$---i.e.\ gates of the form $\op U_{ij}(\theta) = \exp(\mathrm i \theta \op X_{ij})$ for $\theta \in [0, 2\pi)$ and where $\op X_{ij}$ generates rotations in the subspace spanned by $\ket i$ and $\ket j$, we can iteratively apply
\begin{align*}
    \ket0 &\xmapsto{\op U_{01}(\theta_1)} \cos\theta_1 \ket0 + \sin\theta_1 \ket 1 \\
    &\xmapsto{\op U_{12}(\theta_2)} \cos\theta_1 \ket0 + \sin\theta_1 (\cos\theta_2 \ket 1 + \sin\theta_2\ket 2) \\
    &\quad\ \ \ \vdots
\end{align*}
with a suitable sequence of angles $\theta_i$.
Yet implementing arbitrary single qubit rotations again requires non-Clifford gates, and compiling qudit operations down to the $\lceil \log_2 g \rceil$-sized qudit register that $\ket g_G$ resides on likely requires even more such resources.

\begin{figure}[t]
    \centering
    \raisebox{9mm}{\begin{quantikz}[row sep={1cm,between origins},column sep=3mm]
        \qw & \ctrl1 & \targ{} & \qw & \qw & \qw \\
        \qw & \gate{\sqrt{\op X}} & \ctrl{-1} & \ctrl1 & \targ{} & \qw \\
        \qw & \qw & \qw & \gate{\sqrt{\op X}} & \ctrl{-1} & \qw 
    \end{quantikz}}
    \hfill
    \begin{quantikz}[row sep={1cm,between origins},column sep=3mm]
        \qw & \ctrl1 & \qw \\
        \qw & \gate{\sqrt{\op X}} & \qw \\
    \end{quantikz}
	\hspace{2mm}
    \raisebox{3mm}{=}
    \raisebox{1mm}{\begin{quantikz}[row sep={1cm,between origins},column sep=3mm]
        \qw & \qw & \gate{\op T} & \ctrl1 & \qw & \ctrl1 & \qw & \qw \\
        \qw & \gate{\op H} & \gate{\op T} & \targ{} & \gate{\op T^\dagger} & \targ{} & \gate{\op H} & \qw \\
    \end{quantikz}}
    \caption{Left: $\sqrt{\mathrm{SWAP}}$ chain used to create a $\ket g_G$ state, as described in \cref{sec:g-state}. Right: implementation of $\sqrt{\mathrm{CNOT}}$ using 3 $\op T$-gates.}
    \label{fig:g-prep}
\end{figure}
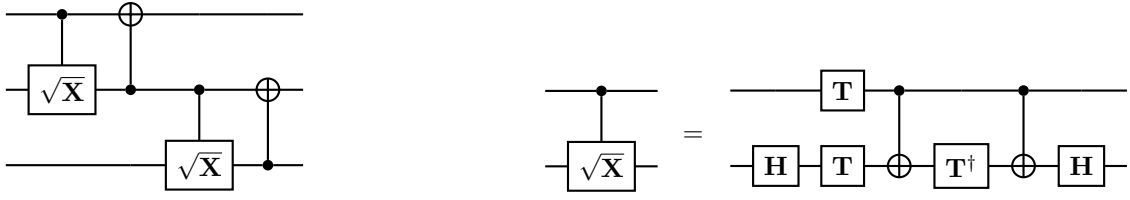

An easier way of creating an amplitude gradient state for $g=2^i-1$ is if we allow ourselves to sacrifice a single dimension as slack space, create a state
\begin{equation}\label{eq:slack-gradient-state}
    \frac{1}{2^g} \left( 2^{(g-1)/2} \ket 0 + 2^{(g-2)/2} \ket 1 + \ldots + \sqrt2 \ket{g-1} + \ket{g} + \ket{g+1} \right) 
\end{equation}
and ignore the slack dimension $\ket{g+1}$ in all subsequent operations; to see that this does not introduce an additional error, note that the entire procedure simply acts as if the amplitudes $\bitalpha_i$ are really given to $g+1$ bits of precision---where the last bit has twice the weight as the second-to-last---but it just so happens that all $\bitalpha_{i,g+1}=0$.

A state like \cref{eq:slack-gradient-state} is readily prepared if we have access to a $g+1$-sized ancilla register; note that we do indeed have access to a $g$-sized space already, since whenever we apply the gate $\op G$, the register storing the amplitudes $\bitalpha_i$ is reset to $\ket0^{\otimes g}$.
In this case, with the aid of successive $\sqrt{\mathrm{CNOT}_{12}}$-gates and $\mathrm{CNOT}_{21}$, and e.g.\ for $g=4$, we perform the operation
\begin{align*}
    \ket{1000}
    &\longmapsto \frac{1}{\sqrt8} \left( \sqrt4 \ket{1000} + \sqrt4 \ket{0100} \right) \\
    &\longmapsto \frac{1}{\sqrt8} \left( \sqrt4 \ket{1000} + \sqrt2 \ket{0100} + \sqrt2 \ket{0010} \right) \\
    &\longmapsto \frac{1}{\sqrt8} \left( \sqrt4 \ket{1000} + \sqrt2 \ket{0100} + \ket{0010} + \ket{0001} \right).
\end{align*}
We then convert the unary representation into a binary representation in the $\ket g_G$-state, e.g.\ $\ket{0100}\ket{00}_G \mapsto \ket{0100}\ket{10}_G$, which is straightforward with a sequence of hard-coded CNOT gates.
To finally unentangle the ancilla register, we apply half of the permutation network in \cref{fig:permutation} which results in a state $\ket{0001}\ket g_G$, and reset the remaining $\ket1$ to $\ket0$.

The Toffoli cost of this amplitude gradient state preparation protocol is thus equal to half the Toffoli cost of a $g+1$-bit permutation network; and $g$ $\sqrt{\mathrm{CNOT}_{12}}=\op H_2 \sqrt{\op{cZ}_{12}} \op H_2 = \op H_2 \op{cS}_{12} \op H_2$, where $\op{cS}$ is a controlled-$\op S$ gate that can be implemented as $\op{cS}_{12} = (\op T_1 \otimes \op T_2) \mathrm{CNOT}_{12} \op T^\dagger_2 \mathrm{CNOT}_{12}$.
The procedure is shown in \cref{fig:g-prep}.

We also note that existing near-term quantum devices sometimes feature a variant of $\mathrm{SWAP}^\theta$ gates natively, e.g.\ through the $\mathrm{fSim}$ gate \cite[Eq.~53]{Google2019}, and the above sequence of successively moving the single $\ket1$-ancilla to the right with a $1/\sqrt2$ factor can simply be implemented directly using $g$ $\sqrt{\mathrm{SWAP}}$ operations.

\section{Optimized Bootstrapping}\label{sec:bootstrapping}
The amplitude amplification overhead from \cref{eq:fp-aa} and the runtime bound in \cref{eq:scaling} indicate that even if our aim was to load a uniform distribution, there would be a $L_1L_2 \approx \sqrt N$ cost; this is not an artefact of our construction, but an overhead already present in \citeauthor{Grover2000}'s original black-box state loading paper.
But we already start with a ``uniform'' initial state, $\ket s$---so where does this overhead originate?

The point to look at is the gradient state $\ket g$ from \cref{eq:gradient-state}, which doesn't serve as a good initial state for all possible amplitude vectors to be loaded.
In fact, it weighs every bit as equally likely to occur in the final state to be loaded.
What we should choose instead is a gradient state representing the most likely bit configuration to occur, weighted by the significance of the bit.
For instance if $N=64$ and $\alpha_i = 1/8$ for all $i$, then $\bitalpha_i = \texttt{0.001}$, and the best initial vector thus consists of a single ``on'' bit, namely the $1/8$\textsuperscript{th} position $\ket{\alpha,0} = 0\ket0 + 0\ket1 + \ket2$.
In this example no amplitude amplification round is necessary at all, since $\ket s$ is already a uniform superposition.

In brief, what we need to choose is an initial state $\ket{\beta} = \sum_{j=0}^{g-1} \beta_j \ket j$ such that the overlap with the intermediate target state is maximized, i.e.
\[
    \sum_{i=0}^{N-1} \bra{\beta} \sum_{j=0}^{g-1} 2^{-j/2} \bitalpha_{ij} \ket j = \sum_{j=1}^g \beta_j 2^{-j/2} \sum_{i=0}^{N-1} \bitalpha_{ij} \eqqcolon \sum_{j=1}^g \beta_j \bar\bitalpha_j.
\]
This is the case if $\ket\beta$ is chosen parallel to the vector defined by the coefficients $\bar\bitalpha_j$ labelling the average $j$\textsuperscript{th} bit weight.
We thus set
\begin{equation}\label{eq:new-initial-state}
    \ket{\bar\bitalpha} \coloneqq \frac{1}{\| \bar\bitalpha \|_2} \sum_{j=0}^{g-1} \bar\bitalpha_j \ket j.
\end{equation}
As an example, consider the case where we aim to load a quantum state with eight coefficients to three bits of precisison, and these coefficients are
\[
\mathtt{0.100},\ \ \mathtt{0.100},\ \  \mathtt{0.100},\ \ \mathtt{0.100},\ \  \mathtt{0.111},\ \ \mathtt{0.101},\ \  \mathtt{0.110},\ \ \mathtt{0.110}.
\]
Then the first bit position is "on" in $8/8$ cases, the second in $3/8$ cases, and the third in $2/8$ cases; thus $\bar A_0 = 1, \bar A_1 = 0.375$, and $\bar A_2=0.25$. The normalization in \cref{eq:new-initial-state} then renders this a valid quantum state.

With this new initial state $\ket{s'} = N^{-1/2} \sum_{i=0}^{N-1} \ket i \ket{\bar\bitalpha}$ instead of \cref{eq:initial-state}, we obtain a shorter $L_1'$ time to amplify towards the intermediate target state $\ket\omega$, namely due to
\begin{align}
    \sqrt{\lambda_1'} = \braket{\omega}{s'} &= \frac{1}{\sqrt N} \frac{1}{\sqrt{\| \bitalpha \|_1}} \frac{1}{\| \bar\bitalpha \|_2} \sum_{i=0}^{N-1} \sum_{j=0}^{g-1} 2^{-(j+1)/2} \bitalpha_{ij} \bar\bitalpha_j  \nonumber\\
%    &= \frac{1}{\sqrt N} \frac{1}{\sqrt{\| \bitalpha \|_1}} \frac{1}{\| \bar\bitalpha \|_2}
%    \sum_{j=0}^{g-1} \left( \sum_{i=0}^{N-1} 2^{-(j+1)/2} \bitalpha_{ij} \right)\left( \sum_{k=0}^{N-1} 2^{-(j-1)/2} \bitalpha_{kj} \right)  \nonumber\\
    &= \frac{1}{\sqrt N} \frac{1}{\sqrt{\| \bitalpha \|_1}} \frac{1}{\| \bar\bitalpha \|_2} \underbrace{\sum_{j=0}^{g-1} \bar\bitalpha_j \bar\bitalpha_j}_{\equiv \| \bar A\|_2^2}  \nonumber\\[-6mm]
    &= \frac{1}{\sqrt N} \frac{\| \bar\bitalpha \|_2}{\sqrt{\| \bitalpha \|_1}}.
    \label{eq:lambda1'}
\end{align}
In conjunction with the unaltered second state loading stage described in \cref{sec:algorithm,eq:lambda_2-eq1}, we have
\begin{align*}
    \lambda_1' \lambda_2 &\ge \frac{1}{N} \frac{\| \bar\bitalpha \|_2^2}{\| \bitalpha \|_1} \frac{ \| \bitalpha \|_2^2 }{ \| \bitalpha \|_1}
    \ge \frac{\| \bar\bitalpha \|_2^2}{N} \frac{ \| \alpha-2^{-g} \|_2^2 }{ \|\alpha\|_1^2 } \\
    %&\ge \frac{\| \bar\bitalpha \|_2^2}{N} \frac{ 1 - 2 \times 2^{-g}\|\alpha\|_1 }{ \|\alpha\|_1^2 }
    % just drop \| \alpha \|_1 < 1
    &\ge \frac{\| \bar\bitalpha \|_2^2}{N} \frac{ 1 - 2 \times 2^{-g} \| \alpha \|_1 }{ \|\alpha\|_1^2 }
    = \frac{1-2^{1-g}\| \alpha \|_1}{N} \left(\frac{\| \bar\bitalpha \|_2}{ \|\alpha\|_1 }\right)^{\!\!2}.
\end{align*}
We thus obtain an optimized state loading protocol with runtime
\begin{equation}\label{eq:scaling-2}
    L' = L_1' L_2 \le \log(2/\delta_1)\log(2/\delta_2)(1+2^{1-g}\| \alpha \|_1) \times \sqrt N \frac{\|\alpha\|_1}{ \| \bar\bitalpha \|_2 }.
\end{equation}

For a discretized set of amplidutes like $\bitalpha$, \cref{eq:lambda1'} gives the exact expression for $\sqrt{\lambda_1'}$.
Nevertheless, the quantity $\| \bar \bitalpha \|_2$---i.e.\ the $2$-norm of the root of the average bit weight---appears hard to get a handle on, in general.
Yet we emphasise that in the black-box state loading picture where we already assume the existence of the oracle unitary $\op U_\mathrm{amp}$ as in \cref{eq:Uamp} anyways (or $\op U_\omega$ as in \cref{eq:Uw}), with the aid of which estimating $\bar\bitalpha_j$---i.e.\ the average value of the $j$\textsuperscript{th} bit---is straightforwardly done by executing the oracle a few times and performing single qubit measurements.
Hence for our purposes we assume that the coefficients $\bar A_j$ and thus also $\| \bar \bitalpha \|_2$ are empirically known to sufficient accuracy; we analytically calculate or estimate the scaling of this quantity for various distributions in \cref{sec:explicit}.

Yet even if we know the coefficients $\bar A_j$ to some precision by an initial sampling procedure, how difficult is it to prepare this new initial state? We first note that the number of \emph{coefficients} of $\ket{\bar A}$ equals the number of \emph{bits of precision} of $\ket A$, i.e.\ $g$; as such, if e.g.\ $g = \BigO(\log N)$, the cost of preparing $\ket{\bar A}$ is negligible (i.e.\ an at most polylogarithmic overhead in $N$) in comparison to the speedup of using this optimized bootstrapping routine, which as per \cref{tab:speedups} can be polynomial or even exponential in $N$.
In turn, the precision requirements on the coefficients $\bar A_j$ are similarly mild: in the step prior to \cref{eq:new-initial-state}, an approximate average state $\ket{\bar B}$ would introduce a relative error in our runtime estimate $L'$ of magnitude
\[
	\frac{L'_{\bar B}}{L'_{\bar A}} = \frac{\| \bar A \|_2\| \bar B \|_2}{\braket{\bar A}{\bar B}}
\]
which shows that knowledge of the coefficients $\bar A_j$ to $g$ bits of precision is sufficient for a slowdown of at most $\BigO(1)$.

We summarise the obtainable speedups when using optimised initial states for a series of amplitude vectors in \cref{tab:speedups}.
\begin{table}[t]
    \hspace*{-1mm}\begin{minipage}{17cm}
    \begin{tabular}{rl | rrr | rrr l }
        \cmidrule[\heavyrulewidth]{1-8}
        & & \multicolumn{3}{c|}{$L$ for $N=\ldots$} & \multicolumn{3}{c}{$L'$ for $N=\ldots$} & \\
        distribution & $L'=\BigO(\cdot)$ & $10^2$ & $10^4$ & $10^6$ & $10^2$ & $10^4$ & $10^6$ & \\
        \cmidrule{1-8}
        delta & $\sqrt N$ & 10 & $100$ & $10^3$  & 10 & $100$ & $10^3$  & \\
        uniform & $1$ & $\vdots\ $ & $\vdots\ \ $ & $\vdots\ \ \,$  &1 &1 &1  & \\
        triangle $\alpha_i\propto i$ & $1$ & $\vdots\ $ & $\vdots\ \ $ & $\vdots\ \ \,$  &1 &1 &1  & \\[3mm]
        powerlaw$^\dagger$ $\alpha_r\propto r^{-k}$ (\Cref{fig:powerlaw})  & $N^{\delta(k)}$ with $\delta(k) < 1/2$ & $\vdots\ $ & $\vdots\ \ $ & $\vdots\ \ \,$  & $\displaystyle\left\{\begin{tabular}{r}
        2 \\
        3 \\
        4
        \end{tabular}\right.\hspace*{-2.5mm}$ & $\displaystyle\begin{tabular}{r}
        4 \\
        9 \\
        25
        \end{tabular}\hspace*{-1.8mm}$ & 
        $\displaystyle\begin{tabular}{r}
        9 \\
        27 \\
        205
        \end{tabular}\hspace*{-1.8mm}$ & 
        $\displaystyle\begin{tabular}{rl}
        \small $k$&\!\!\!\!\!\!\small=\,1/2 \\
        \small $k$&\!\!\!\!\!\!\small=\,3/2 \\
        \small $k$&\!\!\!\!\!\!\small=\,2
        \end{tabular}$ \\[9mm]
        normal$^\dagger$ $\mathcal N(0, \sigma)$ (\Cref{fig:normal}) & $\displaystyle\begin{cases}
        1 & N\lessapprox\sigma \\
        \sqrt{N/\sigma} & \text{otherwise}
        \end{cases}$  & $\vdots\ $ & $\vdots\ \ $ & $\vdots\ \ \,$  &
        $\displaystyle\left\{\begin{tabular}{r}
        1 \\
        1 \\
        5
        \end{tabular}\right.\hspace*{-2.5mm}$ & $\displaystyle\begin{tabular}{r}
        1 \\
        3 \\
        44
        \end{tabular}\hspace*{-1.8mm}$ & 
        $\displaystyle\begin{tabular}{r}
        1 \\
        22 \\
        440
        \end{tabular}\hspace*{-1.8mm}$ & 
        $\displaystyle\begin{tabular}{rl}
        \small\!$\sigma$&\!\!\!\!\!\!\small=\,$10^4$ \\
        \small\!$\sigma$&\!\!\!\!\!\!\small=\,$100$ \\
        \small\!$\sigma$&\!\!\!\!\!\!\small=\,1
        \end{tabular}$\\[9mm]
        random$^\dagger$ $\alpha_i \sim \operatorname{unif}(0,1)$ (\Cref{fig:random}) & 1 & $\vdots\ $ & $\vdots\ \ $ & $\vdots\ \ \,$  & 1&1&1 \\
        \cmidrule[\heavyrulewidth]{1-8}
    \end{tabular}
    \end{minipage}
    \caption{Speedups due to optimized initial state for state loading, as compared to the generic $\sqrt N$ runtime. Explicit calculations are given in \cref{sec:explicit}; for those marked with $^\dagger$, runtimes are conjectured based on empirical analysis and unlikely to be tight.}
    \label{tab:speedups}
\end{table}
Particularly noteworthy is that states not even \emph{that} close to uniform---such as for a triangular set of coefficients where $\alpha_i \propto i$ with a trivial oracle unitary $\op U_\mathrm{amp} = \1$)---can still be simple to load; they require $\BigO(1)$ amplitude amplification rounds .
Furthermore, if the coefficients themselves are unknown (e.g.\ for a random oracle) but at least their expected distribution is known, similar statements of the runtime \emph{in expectation} can be found.
For instance for coefficients sampled uniformly at random from the unit interval (modulo normalisation) we empirically find a $\BigO(1)$ scaling for state loading as well.

We note that the number of amplitude amplification rounds given by the $L'$-times from \cref{tab:speedups} for the powerlaw distribution appears to match up with the ``search with advice'' protocol by \citeauthor{Montanaro2018} \cite[Prop.~3.4]{Montanaro2018}, in the sense that if we used our state preparation protocol to prepare an advice state, then the ``search with advice'' runtime $T$ empirically satisfies $T\times L' = \BigO(\sqrt N)$.

\section{Generalised State Loading}\label{sec:general}
As already mentioned in the introduction, we note that our black-box state loading technique really loads a state
\[
    \ket{\psi} = \sum_{i=0}^{N-1} \left( \sum_{j=1}^g b_{ij} w_{\!j} \right) \ket i
\]
for some boolean matrix $\op B = ( b_{ij} )$ and weight vector $ \op w = (w_{\!j})$, where so far we simply assumed that $w_{\!j} = 2^{-j}$---which is ideal for loading a binary representation of the coefficients $\alpha_i \equiv \sum_j b_{ij} w_{\!j} = \sum_j \bitalpha_{ij} 2^{-j}$.
Yet we can consider the more general case where the $w_{\!j}$ are arbitrary weights, e.g.\ distances in a weighted graph for producing a state to sample from a Travelling Salesman instance, or any other task where fixed scores $w_{\!j}$ are assigned (or not assigned) to individual items $i$.

The state loading procedure is the same as in \cref{sec:algorithm} but for a change in the amplitude gradient state created with a unitary $\op G$, which we replace with a unitary
\[
    \op W \ket 0  =  \sum_{j=0}^{g-1} \sqrt{w_{\!j}} \ket j.
\]
Naturally, optimised bootstrapping as explained in \cref{sec:bootstrapping} can be applied in this context as well.

\section{Conclusion}

We derive runtime bounds for this optimised state loading protocol, and evaluate them analytically---or empirically, where an analytic evaluation is difficult---for a set of widely-used distributions. We find significant speedups as compared to agnostic black box state loading: e.g.\ if the amplitudes follow a powerlaw distribution $\propto r^{-k}$ over $10^6$ elements, agnostic black box state loading would require $\sim 10^3$ amplification rounds, irrespective of the powerlaw's falloff exponent $k$. For $k=2$---where most of the probability mass is concentrated on a few elements---we still only require $\sim 440$ amplification rounds; if the exponent was $1/2$, it would be only $9$.

These speedups in terms of rounds of required amplitude amplification are possible even if we utilize the same oracle as \cite{Sanders2019,Grover2002}. Yet if we have access to an oracle that can query individual bits of the amplitudes to be loaded, obtaining amplitudes with $64$ bits of precision is possible with 6 ancillas; whereas \cite{Sanders2019}'s protocol would require 66 (cf.~\cref{tab:precision}). We again emphasize that this comparison between algorithms that query fundamentally different oracles is not meant to imply that the oracles themselves are comparable, it is a useful comparison in terms of the overall task to be achieved (loading amplitudes in a black box setting), and potentially interesting as a comparison between the usefulness of various oracle variants themselves.

While our technique is only described for non-negative amplitudes, we expect that they can be extended to negative and complex amplitudes (or those given in a different number representation, such as polar coordinates) in a similar fashion to \citeauthor{Sanders2019}, or by utilizing a phase gradient state in conjunction with an oracle can query magnitude and phase of the amplitudes to be loaded directly.

There are further optimisations one could think of.
If, for instance, the amplitudes are all from a fixed set of numbers $\alpha_i \in S$, then it is conceivable that a more efficient number representation than binary can be derived.
Similar to the generalised state loading described in \cref{sec:general}, such a representation would likely improve the runtime further, and require even fewer ancillas.
Finally, an interesting question to pursue would be in what concrete context the generalised black box state loading protocol can be employed.

\section*{Acknowledgements}
We would like to thank Sathya Subramanian, Felix Leditzky and Yuval Sanders for helpful discussions, as well as the reviewers at Quantum for very useful feedback. J.\,B.~is supported by the Draper's Research Fellowship at Pembroke College.

\printbibliography

\newpage
\appendix
\section{Optimized Initial State Calculations}\label{sec:explicit}
We present calculations for the speedups presented in \cref{tab:speedups}, focusing on the dependence of $L'$ given in \cref{eq:scaling-2} in $N$, i.e.
\[
    L' \sim \sqrt N \frac{\|\alpha\|_1}{ \| \bar\bitalpha \|_2 }.
\]
\paragraph{Delta Distribution.}
A single element is marked; hence $\| \alpha \|_1 = \sum_i \alpha_{i=0}^{N-1} = 1$. With an optimized state loading protocol we first need to determine the average bit weights $\bar\bitalpha_j$ as defined just before \cref{eq:new-initial-state}, i.e.
\[
    \bar\bitalpha_j = 2^{-j/2} \sum_{i=0}^{N-1} \bitalpha_{ij}.
\]
As $\bitalpha_{i0}$ denotes the most significant bit, i.e.\ the $1/2$'s, as laid out at the start of \cref{sec:algorithm}, we can assume that for the single marked element $i'$ we have $\bitalpha_{i'j}=1$ for all $j=1, \ldots, g$, and $\bitalpha_{ij}=0$ for all $i\neq i'$.
Then $\bar\bitalpha_j = 2^{-j/2}$, and thus
\[
1 - 2^{-g} \le \| \bar\bitalpha \|_2^2 = \sum_{j=1}^g 2^{-j} \le 1
\]
By \cref{eq:scaling-2} we get $L_1' = \BigO(\sqrt N)$.

\paragraph{Uniform Distribution.}
All elements are marked; hence $\alpha_i=1/\sqrt N$, and $\| \alpha \|_1 = \sum_i \alpha_i = N / \sqrt{N} = \sqrt N$.
For the discrete case we assume there exists one $k$ such that $2^{-k} = 1/\sqrt N$.
For this $k$ we then have $\bitalpha_{ik} = 1$ and $\bitalpha_{ij} = 0$ for all $j\neq k$, and for all $i$. Thus
$\bar\bitalpha_k = 2^{-k/2} \times N = N^{3/4}$ and $\bar\bitalpha_j = 0$ for all $j\neq k$.
Hence $\| \bar\bitalpha \|_2 = N^{3/4}$.
By \cref{eq:scaling-2} we have $L' = \BigO(N^{1/4})$.

This is not yet ideal; we want to obtain a $\BigO(1)$ scaling.
By choosing a different number representation in which we obmit all higher-order bits $j<k$, we obtain $\bar\bitalpha_1 = N / \sqrt 2$ and $\bar\bitalpha_j = 0$ for $j\neq k$; thus $\| \bar\bitalpha \|_2 \sim N$, and the constant runtime follows.

It is worth noting that this particular shortcut is a variant of using the more generalised state loading protocol from \cref{sec:general}, with the weights corresponding to an optimized choice of the distribution's bit representation.

\paragraph{Triangle Distribution.}
We have $\alpha_i \propto i$ for $i=0,\ldots,N-1$, which with normalisation under the 2-norm yields
\[
    \alpha_i = i / \sqrt{N(N-1)(2N-1)/6}.
\]
This means
\[
    \| \alpha \|_1 = \sqrt{\frac32} \frac{N(N-1)}{\sqrt{N(N-1)(2N-1)}} \sim \sqrt{\frac34} \frac{N^2}{N^{3/2}} = \sqrt{\frac34}\sqrt N.
\]
We assume for simplicity that $N$ and the bit representation is chosen such that
\newcommand\mathttt[1]{\texttt{$#1$}}
\[
    \bitalpha_0 = \texttt{0.0}\cdots\texttt{000},\ 
    \bitalpha_1 = \texttt{0.0}\cdots\texttt{001},\ 
    \bitalpha_2 = \texttt{0.0}\cdots\texttt{010},\ 
    \ldots,\ 
    \bitalpha_{N-1} = \texttt{0.1}\cdots\texttt{111}.
\]
This means each bit is equally likely, and is set to one for precisely half of the $N$ amplitudes; it immediately follows that $\bar\bitalpha_j = 2^{-j/2} N/2$, and thus
\[
    \| \bar\bitalpha \|_2 = \frac N2 \sqrt{\sum_{j=1}^g 2^{-j}} = \sqrt{1-2^{-g}} \frac N2.
\]
It follows that $L' = \BigO(1)$.

\begin{figure}[t]
    \hspace*{-1.5cm}\begin{minipage}{18cm}
    \includegraphics[width=6cm]{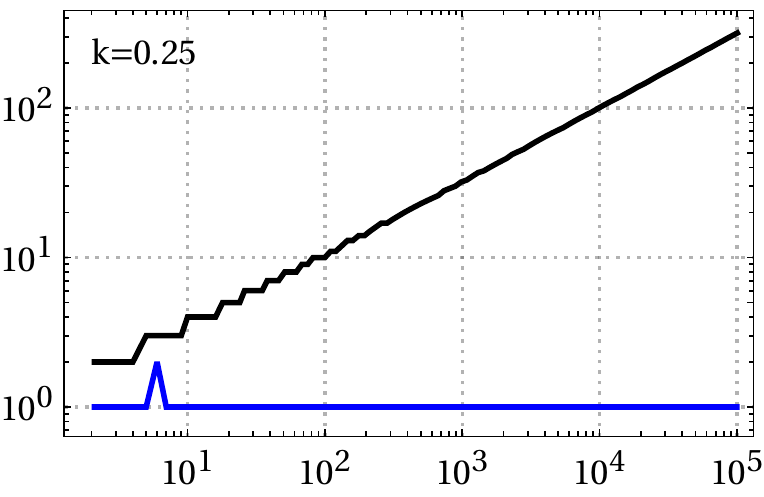}
    \includegraphics[width=6cm]{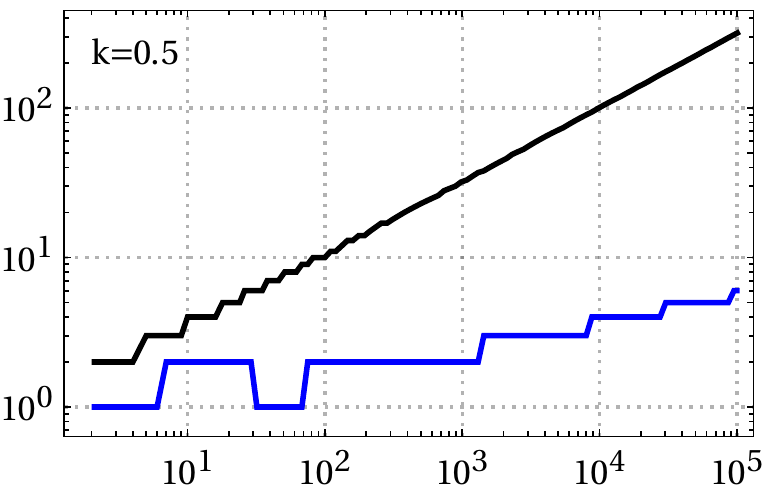}
    \includegraphics[width=6cm]{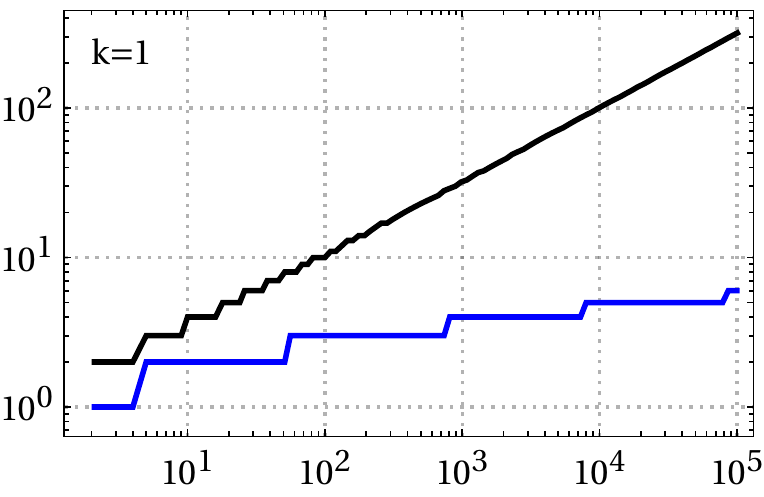}
    
    \includegraphics[width=6cm]{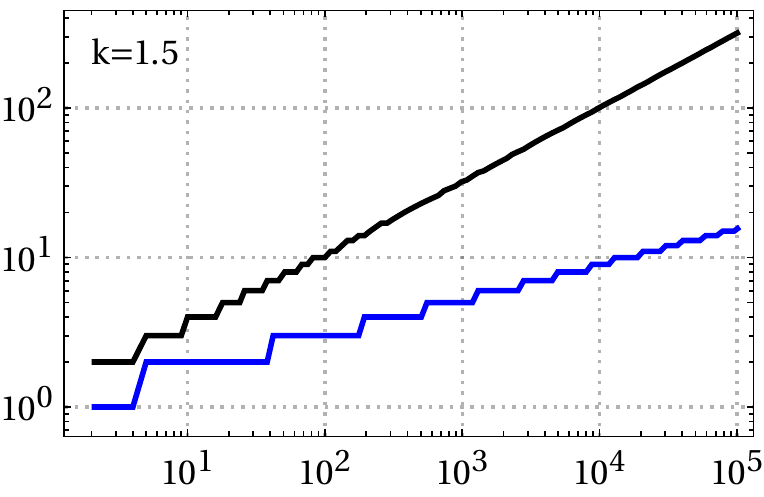}
    \includegraphics[width=6cm]{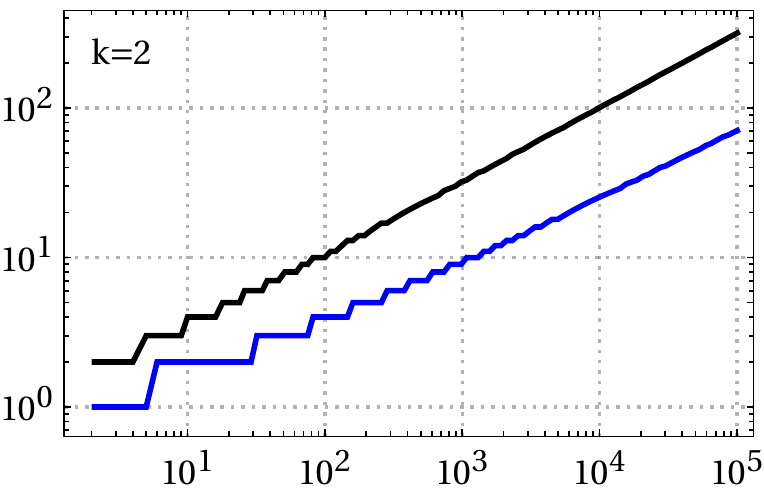}
    \includegraphics[width=6cm]{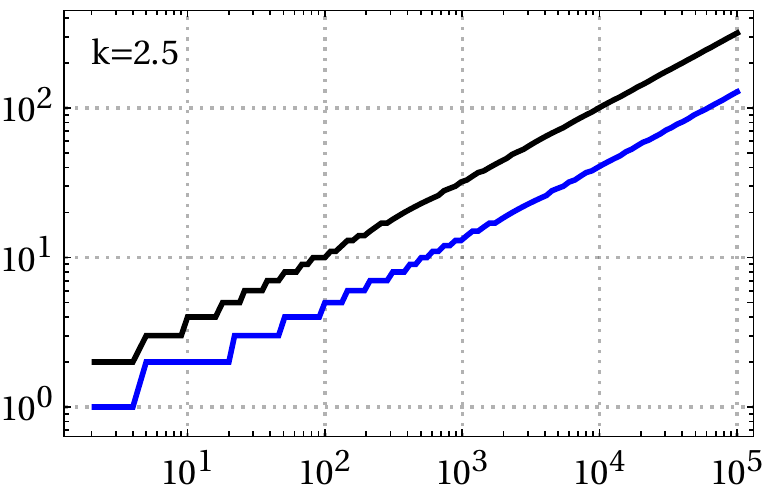}
    \end{minipage}
    \caption{Runtime $L$ (black, from \cref{eq:scaling}) and $L'$ (blue, from \cref{eq:scaling-2}) on vertical axis, evaluated empirically for a powerlaw distribution with various choices of the falloff exponent $k$, and various element counts $N$ on the horizontal axis.
    The kinks where the runtime decreases for increasing $N$ emerge from cases where we shifted the highest-significance bit to more effectively represent the distribution's dynamic range, as done and explained in the case of the uniform distribution.}
    \label{fig:powerlaw}
\end{figure}

\paragraph{Powerlaw Distribution.}
A powerlaw distribution is given by weights $\propto r^{-k}$, for $k\in(0,\infty)$, and where $k$ indicates the rank of the element (if sorted).
As our amplitude vector $\alpha$ is normalised by the $2$-norm, we will assume the task to be to load the amplitudes
\[
    \alpha_r = \frac{r^{-k}}{H_{N,2k}^{1/2}}
    \quad\text{for $r=1,\ldots,N$}
\]
and where the normalisation is given by the Harmonic number $H_{N,k} \coloneqq \sum_{r=1}^N r^{-k}$; then $\| \alpha \|_2=1$ as can be easily verified.
\Cref{fig:powerlaw} shows the resulting runtimes, for various powerlaw falloff choices.

\paragraph{Normal Distribution.}
\begin{figure}[t]
    \hspace*{-1.5cm}\begin{minipage}{18cm}
    \includegraphics[width=6cm]{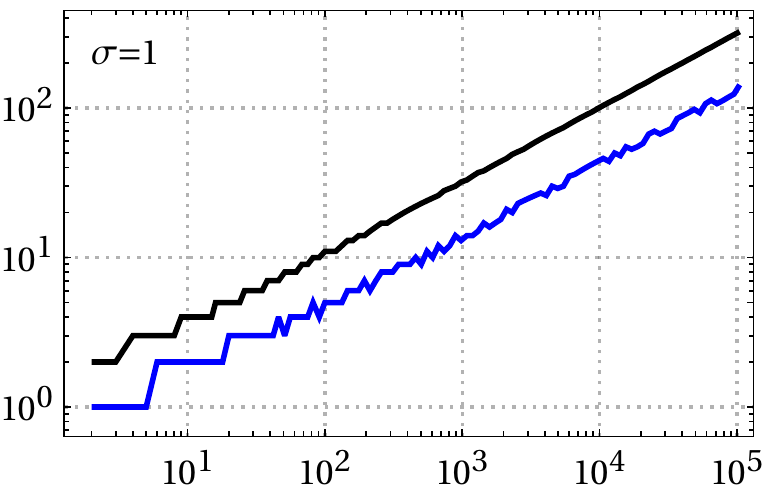}
    \includegraphics[width=6cm]{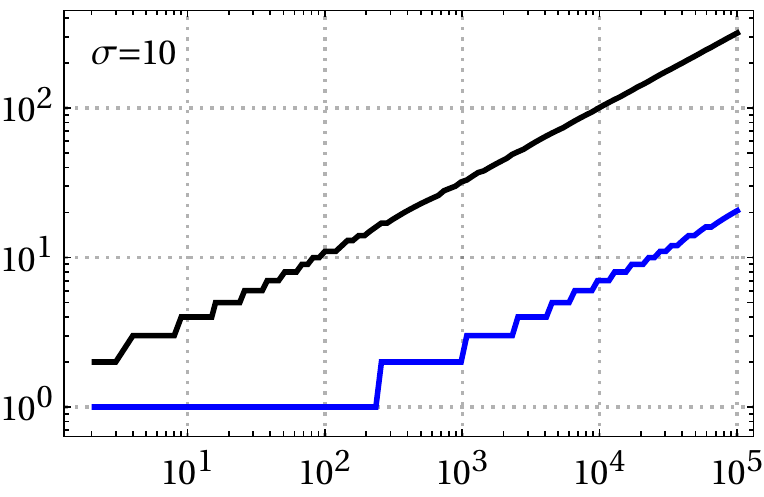}
    \includegraphics[width=6cm]{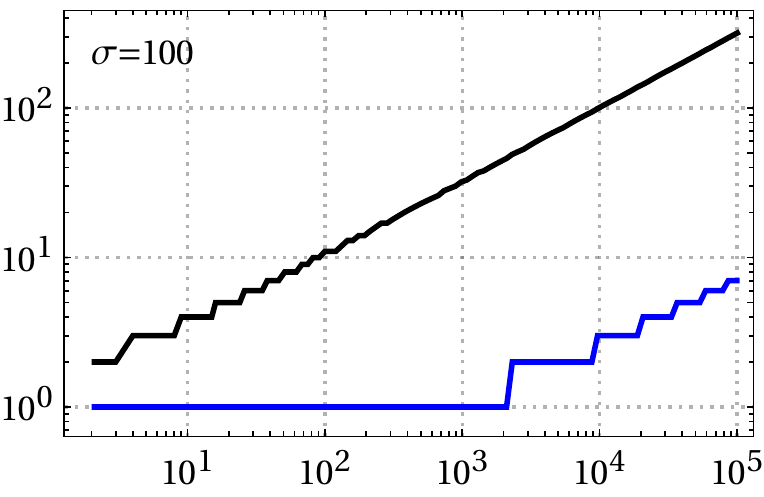}

    \end{minipage}
    \caption{Runtime $L$ (black, from \cref{eq:scaling}) and $L'$ (blue, from \cref{eq:scaling-2}) on vertical axis, evaluated empirically for a normal distribution with various choices of standard deviation $\sigma$, and various element counts $N$ on the horizontal axis.
    The kinks where the runtime decreases for increasing $N$ emerge from cases where we shifted the highest-significance bit to more effectively represent the distribution's dynamic range, as done and explained in the case of the uniform distribution.}
    \label{fig:normal}
\end{figure}

A normal distribution is given by weights $\propto \exp(-x^2/2\sigma^2)$, for $x\in\field R$, and where $\sigma>0$ denotes the standard deviation.
In our discrete case, we simply assume that $\alpha_x$ is proportional to this weight, normalised such that $\| \alpha \|_2=1$.
\Cref{fig:normal} shows the resulting runtimes, for various standard deviation choices.

The scaling suggests that if the standard deviation is large in the context of the number of samples, then the runtime is $\BigO(1)$; otherwise $\BigO(\sqrt{N/\sigma})$.

\paragraph{Random Distribution.}
\begin{figure}
    \centering
    \begin{tikzpicture}
    \node at (0, 0) {\includegraphics{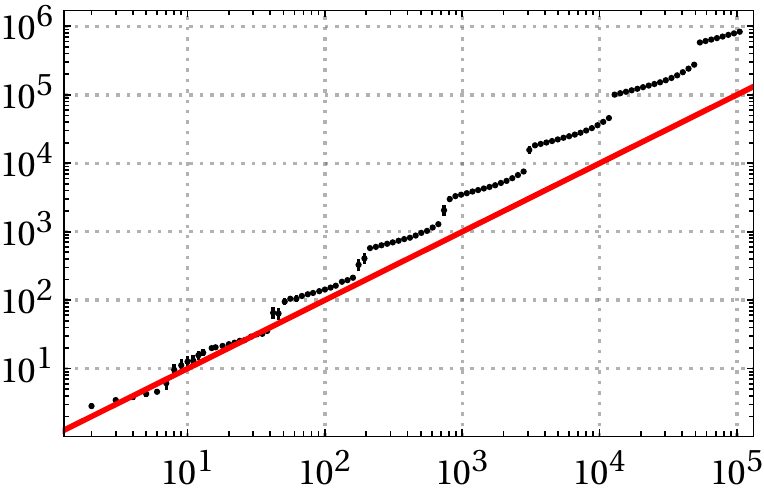}};
    \node at (-4.5, .5) {$\| \bar\bitalpha \|_2$};
    \node at (-2.8, -2.27) {$N=$};
    \end{tikzpicture}
    \caption{2-norm of average bit weight $\| \bar\bitalpha \|_2$ (black) for a random distribution, where amplitudes are drawn uniformly at random from $[0,1]$, and then normalized. As reference in red is the diagonal $N$, indicating that $\| \bar\bitalpha \|_2 = \Omega(N)$.}
    \label{fig:random}
\end{figure}
In case the distribution we wish to load is itself randomly sampled uniformly from the interval $[0,1]$---i.e.\ not with uniform weights, but where the $\alpha_i$ are random samples from $[0,1]$, and the overall vector normalised---one can calculate the runtime (in expectation). If $X_i\sim\operatorname{unif}(0, 1)$ are uniform iid random variables, we have
\[
    \| \alpha \|_1 = \frac{\sum_{i=0}^{N-1} \alpha_i}{\sqrt{\sum_{i=0}^{N-1} \alpha_i}} = \frac{ \sum_i \mathds E( X_i ) }{ \mathds E\left[ \sqrt{\sum_i X_i^2} \right] } \le \frac{N/2}{\sqrt{N}} = \sqrt N/2.
\]
We can check empirically how $\| \bar\bitalpha \|_2$ scales; this is shown in \cref{fig:random}, and the scaling suggests that $\| \bar\bitalpha \|_2 = \Omega(N)$.
As a result, by \cref{eq:scaling-2}, we have $L' = \BigO(1)$.

\paragraph{Sin Distribution.}
For discrete-time random walks, and as mentioned in the introduction, an interesting initial state to prepare is when the $\alpha_i = \sqrt{2/(n+1)} \times \sin(\pi(i+1)/(N+1))$.
A similar argument to the triangle distribution case shows $L'=\BigO(1)$.

\end{document}

%% file: symphony.tex
\usepackage{microtype,dsfont}
\usepackage[dvipsnames]{xcolor}
\definecolor{color1}{RGB}{230,57,70}
\definecolor{color2}{RGB}{29,53,87}
\definecolor{color3}{RGB}{69,123,157}
\usepackage[colorlinks=true,linkcolor=color1,urlcolor=color2,citecolor=color3,anchorcolor=green,breaklinks=true]{hyperref}

\usepackage[capitalize]{cleveref}
\crefname{lemma}{lemma}{lemmas}
\crefname{proposition}{proposition}{propositions}
\crefname{definition}{definition}{definitions}
\crefname{theorem}{theorem}{theorems}
\crefname{conjecture}{conjecture}{conjectures}
\crefname{corollary}{corollary}{corollaries}
\crefname{example}{example}{examples}
\crefname{section}{section}{sections}
\crefname{appendix}{appendix}{appendices}
\crefname{figure}{fig.}{figs.}
\crefname{equation}{eq.}{eqs.}
\crefname{table}{table}{tables}
\crefname{item}{property}{properties}
\crefname{remark}{remark}{remarks}
\crefname{problem}{}{}

\newtheorem{theorem}{Theorem}

\newtheorem{lemma}[theorem]{Lemma}

\usepackage[backend=biber,style=alphabetic,sorting=none,doi=false,isbn=false,url=false,maxbibnames=20,maxcitenames=2]{biblatex}
\renewbibmacro{in:}{%
  \ifentrytype{article}
    {}
    {\printtext{\bibstring{in}\intitlepunct}}}

\newbibmacro{string+doi}[1]{\iffieldundef{doi}{#1}{\href{https://dx.doi.org/\thefield{doi}}{#1}}}
\DeclareFieldFormat{title}{\usebibmacro{string+doi}{\mkbibemph{#1}}}
\DeclareFieldFormat[article]{title}{\usebibmacro{string+doi}{\mkbibquote{#1}}}
\DeclareFieldFormat[incollection]{title}{\usebibmacro{string+doi}{\mkbibquote{#1}}}